\documentclass[conference]{IEEEtran}

\usepackage{cite}
\usepackage{graphicx}
\usepackage{booktabs}
\usepackage{xcolor}
\usepackage{algorithm}
\usepackage{algpseudocode}
\usepackage{float}
\usepackage{tikz}
\usepackage{amsmath,amssymb,amsfonts,amsthm}
\newtheorem{proposition}{Proposition}
\usepackage{hyperref}
\usepackage{url}
\usetikzlibrary{shapes.geometric, arrows.meta, positioning, fit, backgrounds}

\graphicspath{{figures/}}

\newcommand{\finding}[1]{\par\vspace{1.5ex}\noindent\fbox{\parbox{0.95\linewidth}{\textbf{Finding:} #1}}\vspace{1.5ex}\par}
\newcommand{\parh}[1]{\vspace{0.3ex}\noindent\textbf{#1.}\hspace{0.5em}}

\begin{document}

\title{SkillReducer: Optimizing LLM Agent Skills for Token Efficiency}

\author{
\IEEEauthorblockN{Yudong Gao\textsuperscript{1}, Zongjie Li\textsuperscript{1}, Yuanyuan Yuan\textsuperscript{2}, Zimo Ji\textsuperscript{1}, Pingchuan Ma\textsuperscript{3}, Shuai Wang\textsuperscript{1}}
\IEEEauthorblockA{\textsuperscript{1}The Hong Kong University of Science and Technology \quad \textsuperscript{2}Tsinghua University \quad \textsuperscript{3}Zhejiang University of Technology}
\IEEEauthorblockA{\small\texttt{\{ygaodj, zligo, zjiag, shuaiw\}@cse.ust.hk} \quad \texttt{yyyuan@mail.tsinghua.edu.cn} \quad \texttt{pma@zjut.edu.cn}}
}

\maketitle

\begin{abstract}
LLM-based coding agents rely on \emph{skills}, pre-packaged instruction sets that extend agent capabilities, yet every token of skill content injected into the context window incurs both monetary cost and attention dilution.
To understand the severity of this problem, we conduct a large-scale empirical study of 55,315 publicly available skills and find systemic inefficiencies: 26.4\% lack routing descriptions entirely, over 60\% of body content is non-actionable, and reference files can inject tens of thousands of tokens per invocation.
Motivated by these findings, we present \textsc{SkillReducer}, a two-stage optimization framework.
Stage~1 optimizes the routing layer by compressing verbose descriptions and generating missing ones via adversarial delta debugging.
Stage~2 restructures skill bodies through taxonomy-driven classification and progressive disclosure, separating actionable core rules from supplementary content loaded on demand, validated by faithfulness checks and a self-correcting feedback loop.
Evaluated on 600 skills and the SkillsBench benchmark, \textsc{SkillReducer} achieves 48\% description compression and 39\% body compression while improving functional quality by 2.8\%, revealing a \emph{less-is-more} effect where removing non-essential content reduces distraction in the context window.
These benefits transfer across five models from four families with a mean retention of 0.965, and generalize to an independent agent framework.
\end{abstract}


\section{Introduction}
\label{sec:intro}

Large language model (LLM) based coding agents, such as Claude Code~\cite{ClaudeCode2025}, Cursor~\cite{Miller2025Cursor}, and Windsurf, have become essential tools for software development. To customize agent behavior, these platforms support \emph{skills}. A skill is a pre-packaged instruction set containing domain-specific rules, code templates, and reference documentation. The design rationale behind skills is to save tokens: rather than repeating instructions in every conversation, a skill encapsulates reusable knowledge that the agent loads on demand. A growing ecosystem has emerged, with marketplaces like SkillHub hosting curated collections and community forums sharing thousands more.
Skills thus represent an emerging class of \emph{software artifacts}: they are authored, versioned, shared via marketplaces, and maintained by developers, yet they lack the mature optimization ecosystem that traditional source code enjoys.

Ironically, skills as currently authored often increase rather than decrease token consumption. A typical skill consists of two primary functional components: a brief description used by the agent runtime to route user queries, and a main body of instructions that is injected into the context window upon invocation. Because context tokens are expensive, this injection incurs significant costs. At current API pricing, a 10,000-token skill costs \$0.03 to \$0.15 per invocation, and teams with frequent agent usage can spend hundreds of dollars monthly on skill content alone. Our analysis of 55,315 publicly available skills reveals that this cost is systematically mismanaged across both components. At the routing layer, 26.4\% of skills lack descriptions entirely, breaking the routing mechanism and forcing agents to evaluate bodies blindly. Within the bodies, over 60\% of the content consists of non-actionable background or examples. Furthermore, reference-heavy skills can inject tens of thousands of tokens when only a fraction is task-relevant. This phenomenon of \emph{skill bloat}, analogous to software bloat~\cite{Quach2018Debloating}, inflates token costs without improving agent performance.

To address these inefficiencies, we present \textsc{SkillReducer}, a skill debloating framework that systematically removes unnecessary content while preserving functional quality. The design of \textsc{SkillReducer} is directly driven by the bipartite structure of skills, resulting in a two-stage optimization pipeline. Stage~1 targets the routing layer. Since routing equivalence is query-dependent and cannot be verified by a fixed test suite, this stage applies delta debugging~\cite{Zeller2002} with an adversarially constructed oracle to simulate routing decisions. This process compresses verbose descriptions to their 1-minimal form and generates missing ones to restore routing functionality. Stage~2 targets the skill body. Unlike source code, skill content is free-form natural language without syntactic boundaries to guide reduction. Therefore, we adapt program slicing principles using a taxonomy-driven classifier to structurally segment the text. This stage separates actionable core rules, which are always loaded, from examples and background context, which are converted into on-demand modules via progressive disclosure. This restructuring is validated by faithfulness checks and a task-based evaluation featuring a self-correcting feedback loop.

We evaluate our framework on 600 skills alongside the external SkillsBench benchmark, yielding several key findings. \textsc{SkillReducer} achieves significant token savings, with a 48\% mean compression rate for descriptions and a 39\% mean reduction for body tokens. The optimization preserves functional quality effectively, maintaining an 86.0\% pass rate on task-based evaluations and a 100\% pass rate on SkillsBench. Counter-intuitively, the compressed skills improve functional quality by 2.8\% over the originals, suggesting a \emph{less-is-more} effect where removing non-essential content reduces distraction in the context window, particularly for longer and more verbose skills. In a controlled baseline comparison at equivalent token budgets, \textsc{SkillReducer} significantly outperforms LLMLingua~\cite{Jiang2023LLMLingua}, direct LLM compression, truncation, and random removal. Furthermore, the functional benefits transfer across five models from four families and an independent agent framework, confirming the robustness of the optimized skills.

\textbf{Contributions.} In summary, this paper makes the following contributions:
\begin{itemize}
  \item An empirical study of 55,315 skills identifying systemic inefficiencies in the routing and body layers of the current LLM agent ecosystem (Section~\ref{sec:empirical}).
  \item A two-stage optimization framework, \textsc{SkillReducer}, that combines delta-debugging-based description compression with taxonomy-driven progressive disclosure to mitigate skill bloat (Section~\ref{sec:approach}).
  \item A comprehensive evaluation demonstrating substantial token reduction with preserved or improved functional quality, highlighting a less-is-more effect in LLM context management (Section~\ref{sec:evaluation}).
  \item An open-source tool implementing the full debloating pipeline, intended as a build-time preprocessing step for skill authors (to be released upon acceptance).
\end{itemize}

%% ============================================================
\section{Background}
\label{sec:background}

\subsection{LLM Agent Skills}
\label{sec:bg_skills}

Modern LLM-based coding agents support a \emph{skill} mechanism that lets users extend agent behavior with pre-packaged instructions.
A skill $s$ consists of four components.
The \textbf{description} ($s.d$) is a short text used by the agent runtime to decide whether to invoke the skill for a given user request, serving as the routing layer.
The \textbf{body} ($s.b$) is the main instruction document, typically a Markdown file, injected into the agent's context window upon invocation.
\textbf{References} ($s.R$) are optional files such as documentation, code templates, or API specifications loaded alongside the body.
\textbf{Scripts} are executable code that the agent can invoke as tools; they do not consume context tokens in the same way and are out of scope for this work.

When a user issues a request, the agent runtime matches it against all available skill descriptions and selects the best-matching skill.
Upon selection, the skill's body and references are injected into the context window, consuming tokens for the duration of the interaction.

Skills can follow a \emph{monolithic} architecture, where all content resides in a single file and is loaded at once, or a \emph{tiered} architecture, where a compact core is always loaded and supplementary modules are fetched on demand via tool calls such as \texttt{read\_file}.
The tiered pattern follows the principle of \emph{progressive disclosure}~\cite{Nielsen2006}: only essential instructions occupy the context window by default, while examples and background remain available when needed.
Stage~2 of \textsc{SkillReducer} automates the transformation from monolithic to tiered architecture (Section~\ref{sec:stage2}).

\subsection{Context Window Constraints and Prompt Compression}
\label{sec:bg_compression}

LLMs operate within a fixed context window (typically 128K--200K tokens) that must accommodate all inputs: system prompts, skill content, conversation history, and codebase context.
A skill with a large body and references can consume tens of thousands of tokens per invocation, and when multiple skills are active simultaneously, their cumulative cost can dominate the context budget.
Several families of techniques reduce the token footprint of LLM inputs: token-level pruning methods such as LLMLingua~\cite{Jiang2023LLMLingua} remove low-perplexity tokens, embedding-level approaches such as Gisting~\cite{Mu2024Gisting} compress prompts into learned soft tokens, and summarization-based methods use an LLM to directly rewrite content into shorter form.
However, skill content presents a distinct challenge: unlike free-form text, skills mix actionable rules, illustrative examples, and explanatory background within a single document, and compression must preserve both routing correctness and task performance.
These constraints motivate a structure-aware approach rather than uniform token-level or sentence-level pruning.

\subsection{Delta Debugging}
\label{sec:bg_dd}

Delta debugging~\cite{Zeller2002} is a test-case minimization algorithm originally designed to isolate failure-inducing inputs in software.
Given a set of changes $U$ and a monotone predicate $\mathcal{O}$, the \textsc{ddmin} algorithm finds a 1-minimal subset $U^* \subseteq U$ such that $\mathcal{O}(U^*)$ holds and no proper subset of $U^*$ satisfies $\mathcal{O}$.
The algorithm works by binary partitioning: it splits $U$ into halves, tests whether each half alone satisfies $\mathcal{O}$, and recursively discards unnecessary elements.
The result is a minimal set where every remaining element is individually necessary.
Stage~1 of \textsc{SkillReducer} adapts this technique by treating semantic clauses in a routing description as changes and routing correctness as the predicate (Section~\ref{sec:stage1}). This formulation is natural because a skill description, like a failing test case, can be decomposed into independent semantic units, and the goal is to find the smallest subset that preserves routing behavior.

%% ============================================================
\section{Empirical Study}
\label{sec:empirical}

To motivate \textsc{SkillReducer}, we conduct an empirical study across three skill sources:
(1)~\emph{Wild}---55,315 user-created skills crawled from public GitHub repositories (filtered for valid \texttt{SKILL.md} files, deduplicated by content hash);
(2)~\emph{SkillHub}---100 curated skills from the SkillHub marketplace, each with editorial review;
(3)~\emph{Community}---620 skills shared on developer forums and social platforms.
Detailed dataset statistics are provided in Appendix~\ref{app:datasets}.

%% ------------------------------------------------------------
\subsection{Description Quality}
\label{sec:rq1}

The description serves as the primary routing signal: the agent runtime matches user requests against descriptions to select which skill to invoke.
We tokenize all descriptions and bodies using the \texttt{cl100k\_base} tokenizer and report distributions in Figure~\ref{fig:b1_token}.

Our analysis reveals a bimodal quality problem.
On one end, 26.4\% of 55,315 Wild skills have no description at all, and 44.1\% are either missing or under 20 tokens.
These descriptions are still loaded into the router's candidate pool on every invocation, consuming tokens, yet the router can never match them to relevant user requests, so the tokens are wasted on every call.
On the other end, many descriptions that do exist are unnecessarily verbose: SkillHub descriptions average 47.78 tokens, but manual inspection reveals substantial non-routing content such as exhaustive feature lists, redundant trigger phrase enumerations, and usage examples that do not help the router distinguish the skill from competitors.

\begin{figure}[t]
  \centering
  \includegraphics[width=\linewidth]{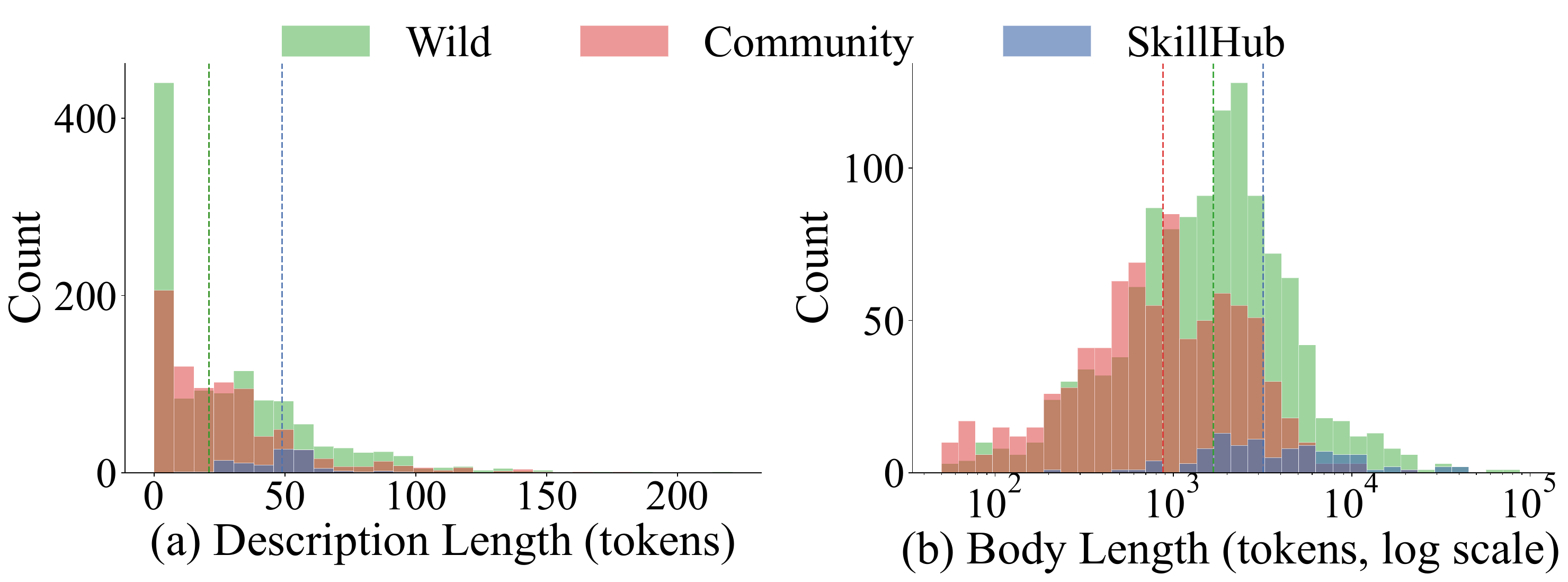}
  \caption{Token-length distributions of descriptions and skill bodies across three sources.}
  \label{fig:b1_token}
\end{figure}

\finding{Descriptions are poorly calibrated: 26.4\% are missing or too short, causing routing failure and wasted token budget; many others contain verbose non-routing content that can be trimmed without affecting routing accuracy.}

%% ------------------------------------------------------------
\subsection{Body Content Taxonomy}
\label{sec:rq2}

Even when correctly routed, a skill's full body is injected into the context window.
Following established content analysis methodology~\cite{Krippendorff2004}, we draw a stratified sample of 90 skills (30 per source), exceeding the saturation threshold of 12--30 samples per stratum identified in prior empirical work~\cite{Guest2006Saturation}.
This yields $n{=}15{,}107$ paragraph-level items, which we classify into five categories: core rule (actionable instructions), background (explanations, rationale), example (code snippets, I/O pairs), template (boilerplate), and redundant (repeated content).
Classification uses an LLM classifier (DeepSeek-V3 at temperature 0) with structured output, followed by a cross-validation pass where each item's label is re-evaluated in context of its neighbors' labels (up to 2 rounds).

Table~\ref{tab:taxonomy} shows the aggregate results: only 38.5\% of items are core rules, while 40.7\% are background and 12.9\% are examples.
To validate this taxonomy, we embed all 15,107 items and apply GMM clustering ($k=5$, selected by silhouette score; see Appendix~\ref{app:gmm} for the model selection analysis).
Figure~\ref{fig:b2_cluster} shows the UMAP projection: five well-separated clusters emerge that align with our LLM-assigned categories, suggesting that the distinction between content types is structurally grounded in the text rather than purely an artifact of the classifier (though we note the silhouette score of 0.393 indicates moderate separation).

\begin{table}[t]
\centering
\caption{Distribution of content types across skill bodies.}
\label{tab:taxonomy}
\resizebox{\columnwidth}{!}{%
\begin{tabular}{lccccc}
\toprule
 & Core Rule & Background & Example & Template & Redundant \\
\midrule
Count & 5,817 & 6,156 & 1,948 & 1,141 & 45 \\
\% & 38.5 & 40.7 & 12.9 & 7.6 & 0.3 \\
\bottomrule
\end{tabular}}
\end{table}

\begin{figure}[t]
  \centering
  \includegraphics[width=0.85\linewidth]{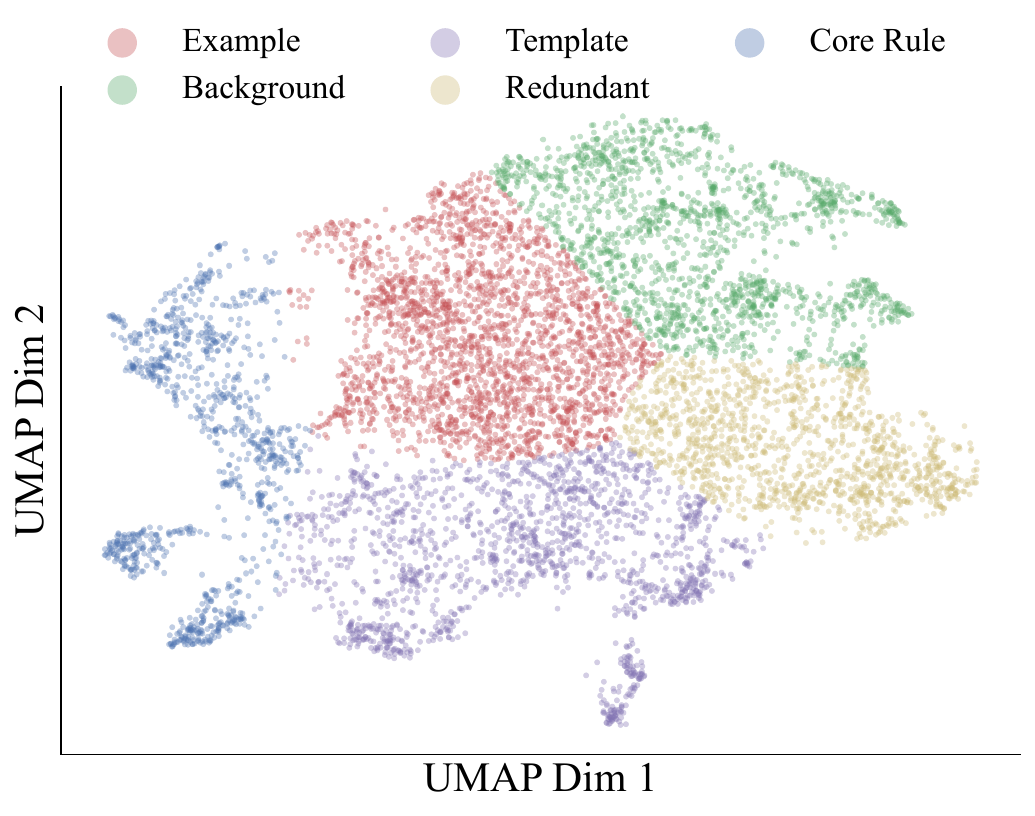}
  \caption{UMAP projection of skill body items with GMM clustering ($k=5$). Five content clusters emerge, suggesting that body content naturally separates into the taxonomy categories (silhouette = 0.393).}
  \label{fig:b2_cluster}
\end{figure}

\finding{Only 38.5\% of skill body content is actionable core rules. Over 60\% is background, examples, or templates that is injected into the context window regardless of task relevance.}

%% ------------------------------------------------------------
\subsection{File Composition}
\label{sec:rq3}

Skills may bundle reference files and scripts alongside the main body.
Among 55,315 Wild skills, 67.5\% consist of a single \texttt{SKILL.md}, but 14.8\% include reference files and 8.3\% include scripts.
For the SkillHub skills with references, the cost is substantial: 100 skills collectively contain 505 files totaling 1.67M tokens, meaning a single skill invocation can inject tens of thousands of reference tokens even when only a fraction is relevant.

\finding{While most skills are single-file, the 14.8\% with references can inject massive token volumes (1.67M tokens across 100 skills) regardless of whether the current task needs them.}

\smallskip
\noindent
These three findings share a common root cause: the absence of separation of concerns in skill authoring.
Authors mix specification (core rules), documentation (background and examples), and data (templates and references) into monolithic files, with no layered architecture to distinguish what the agent must know from what it might need.
Section~\ref{sec:approach} frames this as an automated separation-of-concerns problem and addresses each manifestation with a targeted optimization stage.

%% ============================================================
\section{Approach}
\label{sec:approach}

Guided by the three findings in Section~\ref{sec:empirical}, we propose \textsc{SkillReducer}, a skill debloating framework inspired by established SE practices.
Just as software debloating~\cite{Quach2018Debloating} removes unused code paths to reduce attack surface and binary size, \textsc{SkillReducer} removes non-essential content from skills to reduce token cost.
Table~\ref{tab:requirements} summarizes how each empirical finding drives a specific design requirement.

\begin{table*}[t]
\centering
\caption{From empirical findings to design requirements.}
\label{tab:requirements}
\begin{tabular}{@{}lp{0.88\textwidth}@{}}
\toprule
\textbf{F1} & \textbf{Symptom:} 26.4\% of descriptions missing or too short; others contain verbose non-routing content. \textbf{Consequence:} Too short $\to$ router never selects the skill (tokens wasted on every call); too long $\to$ non-routing filler wastes tokens. \textbf{Requirement:} Minimally sufficient descriptions: generate missing ones; compress verbose ones to only routing-essential content. $\to$ \textbf{Stage~1} \\
\midrule
\textbf{F2} & \textbf{Symptom:} Only 38.5\% of body content is actionable core rules; 60\%+ is background, examples, or templates. \textbf{Consequence:} Non-core content wastes tokens and can distract the agent from core instructions. \textbf{Requirement:} Content-aware separation: core rules always loaded; non-core deferred to on-demand modules. $\to$ \textbf{Stage~2} \\
\midrule
\textbf{F3} & \textbf{Symptom:} Reference files loaded in full regardless of task relevance (up to 1.67M tokens across 100 skills). \textbf{Consequence:} Massive token injection even when only a fraction is needed for the current task. \textbf{Requirement:} Task-dependent loading: deduplicate body--reference overlap; add routing metadata for selective loading. $\to$ \textbf{Stage~2} \\
\bottomrule
\end{tabular}
\end{table*}

The key challenge across all three requirements is that skill content is natural language with implicit structure: unlike code where syntax enables precise slicing, skill bodies interleave rules, examples, and explanations in free-form Markdown.
\textsc{SkillReducer} addresses each finding with a targeted optimization stage (Figure~\ref{fig:overview}).
First, Stage~1 applies delta-debugging-based description minimization to produce minimally sufficient routing descriptions (Section~\ref{sec:stage1}).
Second, Stage~2 employs taxonomy-driven content classification and type-specific compression to separate core rules from non-actionable content (Section~\ref{sec:stage2}).
Finally, Stage~2 further performs cross-file deduplication and reference annotation to eliminate body--reference overlap and enable task-dependent loading (Section~\ref{sec:crossfile}).

Before detailing each stage, we formalize the optimization objective. A skill $s$ consists of a description $s.d$, a body $s.b$, and optional reference files $s.R = \{r_1, \ldots, r_m\}$. The total token cost of invoking $s$ is:
\begin{equation}
  \label{eq:cost}
  \text{Cost}(s) = |s.d| + |s.b| + \sum_{r \in s.R} |r|
\end{equation}
where $|\cdot|$ denotes token count.
\textsc{SkillReducer} produces an optimized skill $s'$ with $\text{Cost}(s') \ll \text{Cost}(s)$ (Eq.~\ref{eq:cost}), subject to two constraints:
(1)~\emph{routing equivalence}: $s'$ is selected for the same queries as $s$;
(2)~\emph{functional retention}: an agent using $s'$ achieves the same task performance as with $s$.

\begin{figure*}[t]
  \centering
  \includegraphics[width=0.94\textwidth]{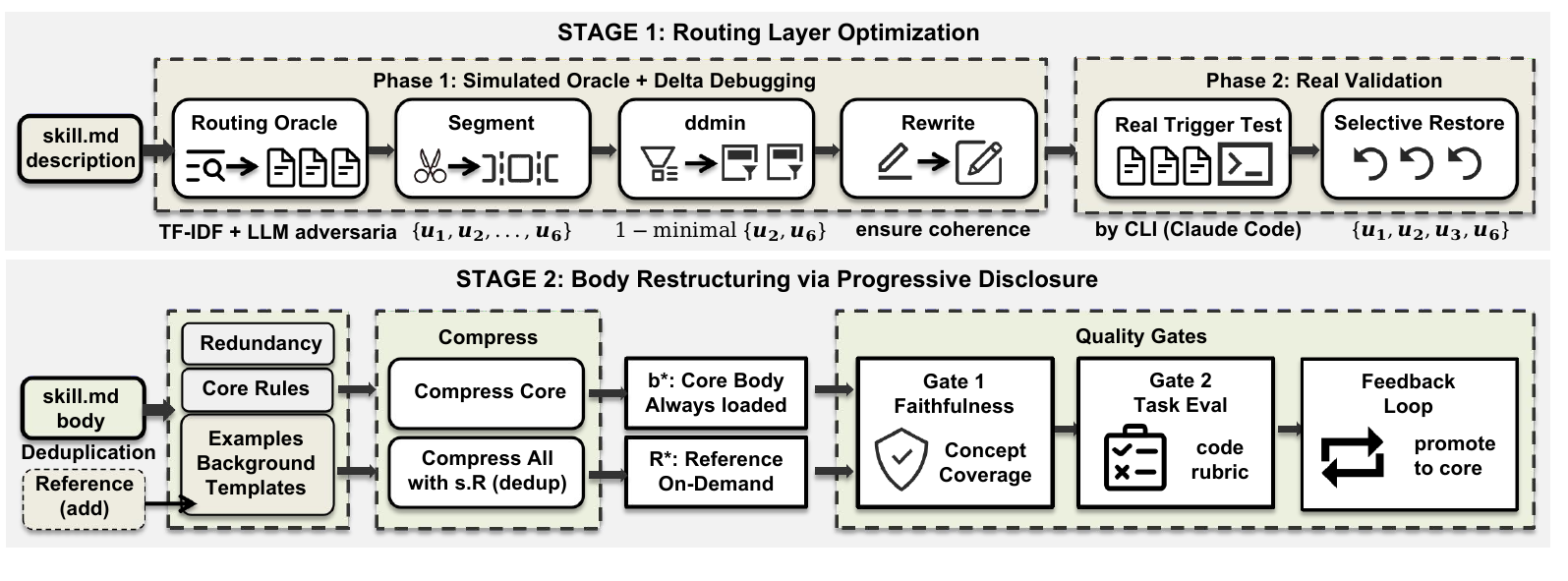}
  \vspace{-8mm}
  \caption{Overview of \textsc{SkillReducer}. Stage~1 (top) compresses routing descriptions via simulated-oracle-driven delta debugging followed by real-environment validation. Stage~2 (bottom) classifies body content, applies type-specific compression, deduplicates references, and validates through faithfulness and task-based quality gates with a feedback loop.}
  \label{fig:overview}
\end{figure*}

%% ------------------------------------------------------------
\subsection{Stage 1: Routing Layer Optimization}
\label{sec:stage1}

The routing layer matches user requests against skill descriptions to decide which skill to invoke.
Finding~1 (Section~\ref{sec:rq1}) reveals a bimodal problem: descriptions that are too short or missing cause routing failure and token waste, while verbose descriptions consume tokens on non-routing content.
Stage~1 addresses both ends: generating descriptions for skills that lack them, and compressing verbose ones to their minimally sufficient form, retaining only the content necessary to distinguish the skill from competitors.
The compression pipeline uses a two-phase design: a fast simulated oracle drives delta debugging, followed by real-environment validation with selective recovery (Algorithm~\ref{alg:stage1}).

\parh{Phase 1: Delta Debugging with Simulated Oracle}\label{sec:stage1_phase1}
Central to Phase~1 is a \emph{simulated routing oracle} $\mathcal{O}(d, Q, C) \to \{0, 1\}$ that tests whether description~$d$ enables correct skill selection.
The oracle takes test queries $Q = \{q_1, \ldots, q_k\}$ and a candidate pool~$C$ containing the target skill, four distractors selected by TF-IDF cosine similarity (name + description), and one adversarial skill $s_{\text{adv}}$.
The distractors are chosen by similarity because in practice, only semantically close skills can cause routing confusion, while dissimilar skills pose no challenge to the router.
To further stress-test the description, we additionally generate $s_{\text{adv}}$ via an LLM: a ``shadow'' skill that is topically similar but functionally distinct, designed to be confusable when the description is vague.
Concretely, the LLM is prompted with the target's name and description and asked to produce a plausible skill in the same domain but with a different purpose (e.g., given a ``JWT authentication'' skill, generating an ``OAuth 2.0 token refresh'' skill); Section~\ref{sec:stage1_results} validates that this adversarial construction meaningfully increases oracle difficulty.
The oracle returns~1 if and only if the routing model consistently selects the target skill for every query in $Q$. To ensure this selection is based on semantic merit rather than positional bias, the order of candidates in $C$ is randomized for each query.

Using this oracle, we apply delta debugging~\cite{Zeller2002} over \emph{semantic clauses} $U = \{u_1, \ldots, u_n\}$, where an LLM segments the description so that each clause captures one coherent routing-relevant concept (e.g., ``JWT authentication'', ``OAuth 2.0 support'').
Unlike sentence-level splitting, which is too coarse (a single sentence may bundle multiple independent signals), or word-level splitting, which is too fine (individual words rarely carry standalone routing meaning), semantic clauses provide a natural granularity that balances precision with tractability.
The \textsc{ddmin} algorithm then follows a partition-and-test process: it recursively splits the clause set, tests whether each subset alone still passes the oracle, and discards unnecessary clauses until the remaining set $U^* \subseteq U$ is 1-minimal, meaning every clause is individually necessary for correct routing.
This process requires $O(n \log n)$ oracle calls rather than $O(2^n)$ for exhaustive search.
Each retained unit is then individually paraphrased to a shorter form (accepted only if routing still passes), followed by a polish step for grammatical coherence.

Notably, the simulated oracle has an inherent limitation: it tests routing in a controlled candidate pool rather than in the real agent runtime, and may therefore be overly optimistic about which clauses can be safely removed.
However, its speed (${\sim}2$s per call) makes it practical as a fast screening mechanism for \textsc{ddmin}'s iterative search.
To address this, Phase~2 performs real-environment validation that catches cases where the simulated oracle is too optimistic.

\parh{Phase 2: Real-Environment Validation}\label{sec:stage1_phase2}
To bridge the gap between simulated and real routing, Phase~2 validates the compressed description by testing whether the actual agent runtime triggers the skill for the same queries.
We deploy the skill with its compressed description into the agent's skill directory and issue each query through the Claude Code CLI, parsing stream events to detect whether the agent invokes the skill.

We first establish a baseline: queries that trigger with the original description form the validation set $Q_{\text{val}} \subseteq Q$.
If the compressed description triggers for all $q \in Q_{\text{val}}$, it is accepted.
Otherwise, we perform selective restore: from the units deleted by \textsc{ddmin}, we greedily add back the unit that most improves the trigger rate, repeating until all queries in $Q_{\text{val}}$ trigger (or a maximum of three restore steps).
If restore fails, the original description is used as fallback.
In Algorithm~\ref{alg:stage1}, \textsc{RealTrigger}$(q, d)$ returns 1 iff deploying description $d$ as a skill and issuing query $q$ through the Claude Code CLI causes the agent to invoke the skill (detected via stream event parsing).

\begin{algorithm}[t]
\caption{Stage 1: Two-Phase Description Compression}
\label{alg:stage1}
\begin{algorithmic}[1]
\Require Skill $s$ with description $s.d$, queries $Q$
\Ensure Compressed description $d^*$
\Statex \Comment{\textbf{Phase 1: Fast compression (simulated oracle)}}
\State $U \leftarrow \textsc{Segment}(s.d)$ \Comment{Split into semantic units}
\State $U^* \leftarrow \textsc{ddmin}(U, \mathcal{O}_{\text{sim}}, Q)$ \Comment{1-minimal subset}
\For{each $u_i \in U^*$} \Comment{Selective rewriting}
  \State Try shorter paraphrase; keep if $\mathcal{O}_{\text{sim}}$ passes
\EndFor
\State $d_{\text{fast}} \leftarrow \textsc{Polish}(\textsc{Join}(U^*))$
\Statex \Comment{\textbf{Phase 2: Validation (real agent trigger)}}
\State $Q_{\text{val}} \leftarrow \{q \in Q \mid \textsc{RealTrigger}(q, s.d) = 1\}$
\If{$\forall q \in Q_{\text{val}}: \textsc{RealTrigger}(q, d_{\text{fast}}) = 1$}
  \State \Return $d_{\text{fast}}$ \Comment{Compression validated}
\EndIf
\State $D \leftarrow U \setminus U^*$ \Comment{Deleted units}
\For{$i = 1$ to $3$} \Comment{Selective restore}
  \State $u^+ \leftarrow \arg\max_{u \in D} \textsc{TriggerRate}(U^* \cup \{u\}, Q_{\text{val}})$
  \State $U^* \leftarrow U^* \cup \{u^+\}$; $D \leftarrow D \setminus \{u^+\}$
  \If{all $Q_{\text{val}}$ trigger}
    \State \Return $\textsc{Join}(U^*)$
  \EndIf
\EndFor
\State \Return $s.d$ \Comment{Fallback to original}
\end{algorithmic}
\end{algorithm}

Section~\ref{sec:stage1_results} evaluates this two-phase design on 600 skills, where the selective restore mechanism achieves 100\% routing preservation across all non-obsolete skills.

\parh{Description Generation}\label{sec:stage1_generate}
The two-phase pipeline above handles skills with existing descriptions. For the 26.4\% of skills with missing or short descriptions ($\leq$40 tokens), we instead generate a description from the skill body.
The LLM extracts three routing signals from the body, each in 20--40 tokens: (1)~primary capability, (2)~trigger condition, and (3)~unique identifiers (library names, API endpoints, etc.).
The generated description is validated through Phase~1's simulated oracle before acceptance; Phase~2 real-trigger validation is applied to a sampled subset (Section~\ref{sec:stage1_results}).

%% ------------------------------------------------------------
\subsection{Stage 2: Body Restructuring via Progressive Disclosure}
\label{sec:stage2}

Even after routing optimization, the full skill body is injected into the agent's context upon invocation.
Our taxonomy analysis (Section~\ref{sec:rq2}) shows that only 38.5\% of body content is actionable.
In practice, most skills follow the monolithic pattern: authors embed reference-like content (examples, templates, background documentation) directly in the body rather than separating it into dedicated reference files.
Stage~2 automatically identifies and separates such content, effectively transforming monolithic skills into the tiered architecture described in Section~\ref{sec:bg_skills}.
Stage~2 restructures each skill from a monolithic document into a \emph{core module} $b^*$ (always loaded) and \emph{on-demand reference modules} $R^* = \{r^*_1, \ldots, r^*_n\}$ (loaded only when the agent explicitly requests them via a \texttt{read\_file} tool call).
The optimized cost becomes:
\begin{equation}
  \label{eq:cost_opt}
  \text{Cost}'(s) = |s'.d| + |b^*| + \sum_{r \in R^*_{\text{used}}} |r|
\end{equation}
where $R^*_{\text{used}} \subseteq R^*$ are only the references the agent actually loads for a given task.
Algorithm~\ref{alg:stage2} describes the full pipeline.

\begin{algorithm}[t]
\caption{Stage 2: Progressive Disclosure}
\label{alg:stage2}
\begin{algorithmic}[1]
\Require Skill $s$ with body $s.b$, references $s.R$
\Ensure Core body $b^*$, reference modules $R^*$
\Statex \Comment{\textbf{Step 1: Content Classification}}
\State $\mathcal{I} \leftarrow \textsc{Classify}(s.b)$ \Comment{Assign each item a type}
\State $I_{\text{core}} \leftarrow \{x \in \mathcal{I} \mid x.\text{type} = \text{core\_rule}\}$
\State $I_{\text{ex}}, I_{\text{tp}}, I_{\text{bg}} \leftarrow$ items by type
\Statex \Comment{\textbf{Step 2: Type-Specific Compression}}
\State $b^* \leftarrow \textsc{CompressCore}(I_{\text{core}})$
\State $r_{\text{ex}} \leftarrow \textsc{Dedup}(I_{\text{ex}})$ \Comment{1 best example/concept}
\State $r_{\text{tp}} \leftarrow \textsc{Dedup}(I_{\text{tp}})$
\State $r_{\text{bg}} \leftarrow \textsc{Summarize}(I_{\text{bg}})$
\Statex \Comment{\textbf{Step 3: Cross-File Deduplication}}
\For{each $r \in s.R$}
  \State $r' \leftarrow \textsc{RemoveOverlap}(r, s.b)$
  \If{$|r'| \geq 30$}
    \State Add $\textsc{Compress}(r')$ to $R^*$
  \EndIf
\EndFor
\State $R^* \leftarrow R^* \cup \{r_{\text{ex}}, r_{\text{tp}}, r_{\text{bg}}\}$ \Comment{Non-empty only}
\Statex \Comment{\textbf{Step 4: Annotate References}}
\For{each $r \in R^*$}
  \State $r.\text{when} \leftarrow \textsc{GenWhen}(r)$ \Comment{Trigger clause}
  \State $r.\text{topics} \leftarrow \textsc{GenTopics}(r)$ \Comment{3--5 keywords}
\EndFor
\Statex \Comment{\textbf{Step 5: Quality Gates}}
\If{$\neg\textsc{Gate1}(s.b, b^*)$} \Comment{Faithfulness}
  \State Rollback failed content types
\EndIf
\For{$i = 1$ to $2$} \Comment{Gate 2 + Feedback}
  \State $\text{ret} \leftarrow \textsc{Gate2}(s, b^*, R^*)$ \Comment{Task eval}
  \If{$\text{ret} = 1.0$}
    \State \textbf{break}
  \EndIf
  \State $P \leftarrow \textsc{Analyze}(\text{failed tasks})$ \Comment{Promote}
  \State $I_{\text{core}} \leftarrow I_{\text{core}} \cup P$; recompress $b^*, R^*$
\EndFor
\State \Return $b^*, R^*$
\end{algorithmic}
\end{algorithm}

\parh{Content Classification}\label{sec:classify}
The first step segments each skill body into paragraph-level items and classifies them into five categories from our taxonomy (Table~\ref{tab:taxonomy}):
core rule (actionable instructions the agent must follow), background (explanations and rationale), example (illustrative code snippets), template (boilerplate for copy-paste), and redundant (repeated content).
Only core rules remain in the always-loaded module; the other categories are separated into on-demand reference modules (\texttt{examples.md}, \texttt{templates.md}, \texttt{background.md}), and redundant items are discarded.
This separation is analogous to \emph{program slicing}~\cite{Weiser1984}: we identify which content items directly affect task execution and defer the rest.
Unlike classical slicing, skill content lacks formal dependency graphs, so we approximate the analysis with a taxonomy-driven LLM classifier and rely on the feedback loop (Section~\ref{sec:feedback}) to empirically recover missed dependencies.
If the LLM fails to produce a valid classification after three retries, the item defaults to core rule as a conservative fallback.

\parh{Type-Specific Compression}\label{sec:compress}
After classification, each content type undergoes tailored compression.
Core rules are merged by semantic similarity and tightened into concise bullet points, with the constraint that the compressed core must be strictly shorter than the input.
Examples are grouped by concept, retaining only the most representative example per concept and stripping comments and boilerplate, which typically reduces example count by 60--70\%.
Templates are deduplicated by concept in the same manner.
Background content is summarized into a concise paragraph while preserving all unique factual claims such as numbers, thresholds, and API endpoints.

\parh{Cross-File Deduplication}\label{sec:crossfile}
For skills with external reference files $s.R$ (14.8\% of skills; Section~\ref{sec:rq3}), we perform overlap detection between each reference $r$ and the full original body $s.b$ (lines~8--12), removing content that is already captured in either the core module or the body-derived reference modules.
An LLM identifies concepts that appear in both, removes duplicated content from $r$, and compresses the remaining unique content.
References that shrink below 30 tokens are discarded as fully redundant.
Each surviving reference is annotated with routing metadata (lines~13--15):
a ``when'' clause describing the trigger condition (e.g., ``you need to write YAML configuration'') and 3--5 topic keywords extracted from the content.
These annotations enable the agent to decide whether to load a reference without reading it first.

\parh{Gate 1: Faithfulness Verification}\label{sec:gate1}
After compression, we verify structural faithfulness (line~16): an LLM is given the original body and the compressed core, and checks whether all key operational concepts are preserved.
Gate~1 operates per content type: each type (core rules, examples, templates, background) is checked independently, and only failing types are rolled back to their original content.
Formally, let $\mathcal{C}_\tau(b)$ denote the set of operational concepts of type $\tau$ in body $b$. Gate~1 requires:
\begin{equation}
  \label{eq:faithful}
  \forall \tau:\; \mathcal{C}_\tau(s.b) \subseteq \mathcal{C}_\tau(b^*) \cup \bigcup_{r \in R^*}\mathcal{C}_\tau(r)
\end{equation}
where the union accounts for concepts moved from the core to reference modules.
This structural check is necessary but not sufficient, since a compressed body may preserve all concepts yet still confuse an agent in practice (see Appendix~\ref{app:guarantee} for the full quality safeguard design).

\parh{Gate 2: Task-Based Evaluation}\label{sec:gate2}
Gate~2 evaluates functional quality by testing whether an LLM agent can still complete real tasks using the compressed skill (lines~18--23).
We generate five diverse evaluation tasks per skill, mixing ``core-only'' tasks (answerable from $b^*$ alone) and ``needs-reference'' tasks (requiring at least one reference module).
Each task includes a rubric of required criteria derived from the original body.

The agent executes each task under three conditions:
\begin{itemize}
  \item \textbf{Condition D} (lower bound): no skill content.
  \item \textbf{Condition A} (baseline): original body $s.b$ and all references $s.R$.
  \item \textbf{Condition C} (compressed): core body $b^*$, with references $R^*$ accessible via \texttt{read\_file}.
\end{itemize}

Under Condition~C, the agent is equipped with a \texttt{read\_file} tool, natively supported by the skill platform's tiered architecture (Section~\ref{sec:bg_skills}), that returns the content of a requested reference module.
The agent autonomously decides which (if any) references to load based on the when/topics annotations, up to a maximum of six tool calls per task.

Responses are scored via two mechanisms depending on task type: code execution tasks (52.3\% of all Gate~2 evaluations) run the agent's output as Python code and verify correctness via assertion-based test scripts (analogous to pytest), while rubric tasks (47.7\%) are scored by an LLM judge against the rubric criteria.
This hybrid design ensures that over half of all assessments are deterministic and LLM-judge-independent.
To further mitigate LLM evaluator bias, we use separate models for compression (DeepSeek-V3) and evaluation (Qwen3.5~\cite{Qwen2025Qwen3}), and validate the LLM judge against the deterministic verifier on 1,646 overlapping tasks (Cohen's $\kappa{=}0.939$; Section~\ref{sec:threats}).
We define the \emph{retention} metric for task $t$ as:
\begin{equation}
  \label{eq:retention}
  \text{Retention}(t) = \begin{cases}
    1.0 & \text{if } \text{score}_A(t) = 0 \\
    \min\!\left(\frac{\text{score}_C(t)}{\text{score}_A(t)},\; 1.0\right) & \text{otherwise}
  \end{cases}
\end{equation}
When $\text{score}_A = 0$ (the original skill already fails), retention defaults to 1.0 since there is no baseline to degrade.
19.8\% of skills are ``ceiling'' cases where $D{=}A{=}C{=}1.0$ (the task is too easy for all conditions); these are excluded in our non-ceiling analysis (Section~\ref{sec:gate2_results}).
A skill passes Gate~2 if $\text{Retention}(t) = 1.0$ for all tasks $t$.

\parh{Feedback Loop}\label{sec:feedback}
If Gate~2 reveals tasks where $\text{score}_C < \text{score}_A$ (lines~22--23), the system analyzes the failed rubric criteria and identifies which non-core items (examples, background, or templates) are needed to satisfy them.
These items are promoted to core rules: their classification is changed to \texttt{core\_rule} and they are included in the always-loaded module.
Crucially, promoted items are appended to the core in their original form without further compression, preserving the content that Gate~2 identified as necessary; only the non-promoted core items retain their compressed form.
Gate~1 is not re-run after promotion because the core set only grows ($I_{\text{core}}^{(i)} \subseteq I_{\text{core}}^{(i+1)}$), so the faithfulness property (Eq.~\ref{eq:faithful}) is maintained.
Gate~2 is then re-run on the expanded core.
This loop executes at most twice; if retention remains below 1.0 after two iterations, the current best compression is used.

This self-correcting mechanism ensures that aggressive compression does not silently degrade task performance: the system trades a small increase in token cost (promoted items) for full functional retention.

The loop is guaranteed to terminate: the core set $I_{\text{core}}^{(i)}$ grows monotonically ($I_{\text{core}}^{(i)} \subseteq I_{\text{core}}^{(i+1)}$) within a finite item set $\mathcal{I}$, reaching a fixpoint in at most $|\mathcal{I}| - |I_{\text{core}}^{(0)}|$ steps (formal proof in Appendix~\ref{app:theory}).
In practice, we cap the loop at two iterations to bound computational cost.

\parh{Running Example}\label{sec:example}
We illustrate the full pipeline on the \texttt{marketing\allowbreak-strategy-pmm} skill (product marketing, positioning, GTM).
Stage~1 compresses the description from 87 to 32 tokens (63\% reduction): \textsc{ddmin} identifies that the trigger-phrase list is redundant because the router can infer triggers from the feature keywords alone.
Stage~2 classifies the 2,543-token body into core rules (KPIs, methodology steps, decision criteria), background (workflow explanations), examples (persona-specific messaging), and templates (HubSpot configuration snippets).
The core is compressed to 540 tokens (79\% reduction), and three on-demand reference modules are created: \texttt{templates.md} (327 tok), \texttt{examples.md} (684 tok), and \texttt{background.md} (602 tok).
End-to-end, invocation cost drops from 12,019 to 540 tokens for core-only tasks, or to 7,231 tokens if all references are loaded.
Gate~2 confirms $\text{score}_C{=}1.0$ vs.\ $\text{score}_A{=}0.93$, demonstrating the less-is-more effect.
A detailed walkthrough with before/after excerpts is provided in Appendix~\ref{app:example}.

%% ============================================================
\section{Evaluation}
\label{sec:evaluation}

We evaluate \textsc{SkillReducer} on 600 skills (87 official from SkillHub, 464 community, and 49 wild sampled from GitHub) and one external benchmark.
Our experiments answer four research questions:
\textbf{RQ1}: How much token reduction does \textsc{SkillReducer} achieve?
\textbf{RQ2}: Does compression preserve functional quality?
\textbf{RQ3}: What components drive the quality preservation?
\textbf{RQ4}: Do the results generalize across models and frameworks?

%% ------------------------------------------------------------
\subsection{Experimental Setup}
\label{sec:setup}

\noindent\textbf{Datasets.}
Unlike the empirical study (Section~\ref{sec:empirical}), which analyzes 55,315 skills for distributional patterns, the evaluation requires running each skill through the full compression and evaluation pipeline, which is substantially more expensive.
We therefore curate a balanced evaluation set from three sources, filtering out low-quality and duplicate content:
(1)~\emph{Official} (87 evaluable skills from the 100 SkillHub skills; 13 excluded due to missing bodies or unparseable content),
(2)~\emph{Community} (464 community-shared skills from the 620 in our empirical study; 156 excluded due to empty bodies or duplicate content), and
(3)~\emph{Wild} (49 skills sampled from the 55,315 GitHub skills, stratified by body length to cover the full size spectrum).
Additionally, we use \emph{SkillsBench} (87 tasks across 229 skills with deterministic pytest verifiers) as an independent external benchmark.

\noindent\textbf{Conditions.} We evaluate each task under three setups: \textbf{Condition D} establishes a lower bound by providing the agent with no skill at all; \textbf{Condition A} serves as our baseline by providing the original, uncompressed skill body and all references; and \textbf{Condition C} tests \textsc{SkillReducer} by providing only the compressed core and loading references on demand.

\noindent\textbf{Metrics.}
We report:
(1)~\emph{compression ratio} = $1 - |b^*|/|s.b|$,
(2)~\emph{pass rate} = fraction of skills where $\text{score}_C \geq \text{score}_A$,
(3)~\emph{improvement rate} = fraction of skills where $\text{score}_C > \text{score}_A$, and
(4)~\emph{retention} (Eq.~\ref{eq:retention}).
All intervals are bootstrap 95\% CIs (10,000 resamples).
Since each skill is evaluated under both conditions A and C (paired design) and scores are bounded in $[0, 1]$ with 19.8\% ceiling cases ($D{=}A{=}C{=}1.0$), we use the non-parametric Wilcoxon signed-rank test for pairwise comparisons, with Cohen's $d$ as an approximate effect size.
The baseline comparison (Table~\ref{tab:baselines}) uses one-sided sign tests due to the high proportion of ceiling ties.

\noindent\textbf{Implementation.}
Stage~1 uses DeepSeek-V3~\cite{DeepSeek2024V3} for description segmentation, compression, and adversarial skill generation, and DeepSeek-R1~\cite{DeepSeek2025R1} as the simulated routing oracle; Phase~2 real-trigger validation runs through the Claude Code CLI~\cite{ClaudeCode2025}.
Stage~2 uses DeepSeek-V3 for content classification, body compression, reference deduplication, and faithfulness verification.
Task generation, agent execution, and evaluation use Qwen3.5~\cite{Qwen2025Qwen3}, an open-source model, under a separate API key and session to prevent information leakage between compression and evaluation.
Section~\ref{sec:rq_crossmodel} further evaluates generalization across five models from four families.
All tokenization uses OpenAI's \texttt{cl100k\_base} encoder via \texttt{tiktoken}.

\noindent\textbf{Scale.}
In total, the evaluation compresses 600 skills through both stages, generates and evaluates 3,000 tasks across three conditions (D/A/C), validates routing on 600 skills in real agent environments, and tests generalization across 5~models from 4~families and an independent agent framework.
Including ablation (50 skills $\times$ 6 conditions), baseline comparisons (4 methods), and wild-skill validation (198 skills), this constitutes one of the largest evaluations of LLM agent skill optimization to date.

%% ------------------------------------------------------------
\subsection{RQ1: Token Reduction}
\label{sec:rq_compression}

\begin{figure}[t]
  \centering
  \includegraphics[width=\linewidth]{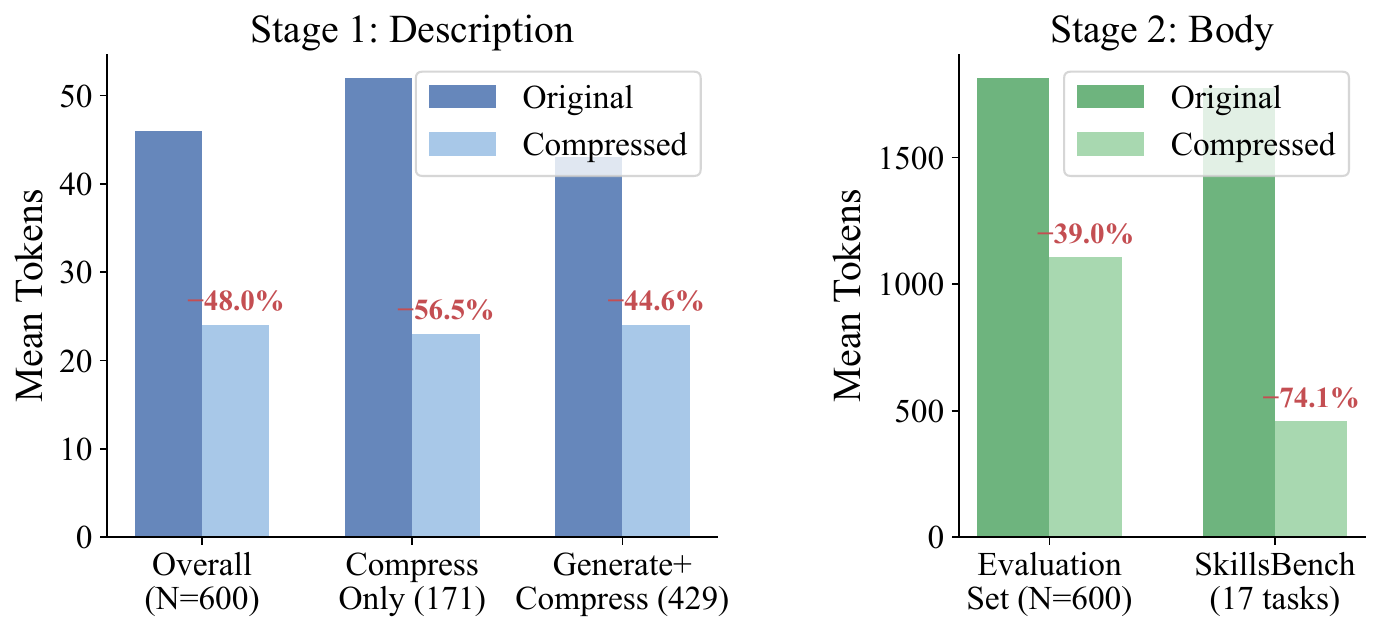}
  \vspace{-6mm}
  \caption{Token reduction achieved by \textsc{SkillReducer}.}
  \vspace{-2mm}
  \label{fig:compression}
\end{figure}

Figure~\ref{fig:compression} summarizes the compression results across both stages. Of the 600 skills, 171 had existing descriptions ($>$40 tokens) that were compressed directly, achieving \textbf{56.5\%} mean reduction.
The remaining 429 had missing or short descriptions ($\leq$40 tokens): we first generate a description from the body, then compress it; the 44.6\% reduction is measured relative to the generated (not original) description, reflecting additional compressibility in LLM-generated text.
Overall, Stage~1 achieves \textbf{48.0\%} mean description token reduction (median 59.0\%, 95\% CI: [45.2\%, 50.8\%]).

Body compression achieves a mean token reduction of \textbf{39.0\%} (median 43.0\%, 95\% CI: [36.2\%, 41.8\%]), saving an average of ${\sim}$1,000 tokens per skill.
On SkillsBench (84 of 87 tasks successfully compressed), body tokens drop from 359K to 84K (\textbf{75.0\%} mean per-task reduction), as these curated skills tend to have longer, more compressible bodies.

\noindent\textbf{Scalability to Wild Skills.}
To validate that compression generalizes beyond our evaluation set, we apply Stage~2 Gate~1 (compression + faithfulness check, no task evaluation) to a stratified sample of 200 skills drawn from the 55,315 Wild dataset (198 successful), stratified by body length and file structure.
The results confirm scalability: mean core reduction is \textbf{77.5\%} (median 81.5\%), with a clear length-dependent pattern: tiny skills ($<$500 tokens) achieve 58.8\% reduction, medium skills (1--3K tokens) achieve 79.1\%, large skills ($>$3K tokens) achieve 89.1\%, and extra-large skills ($>$10K tokens) reach 95.8\%.
Faithfulness verification passes for 74.7\% (148/198) of sampled skills.
The higher compression rates compared to the evaluation set (77.5\% vs.\ 39\%) are partly due to the absence of Gate~2 feedback (which promotes items back to core, reducing compression) and partly due to the composition of wild skills: longer bodies offer more redundancy to remove.

\noindent\textbf{End-to-End Token Savings.}
Using our empirical core ratio $\rho{=}0.383$ and a reference loading rate $p{=}0.3$ (from Gate~2 tool-call logs; see Proposition~\ref{prop:cost} in Appendix~\ref{app:theory} for the formal derivation), the expected body cost is $(\rho + p(1{-}\rho)) \cdot |s.b| = 0.568 \cdot |s.b|$, a 43.2\% reduction before core compression.
An end-to-end analysis confirms mean input savings of \textbf{26.8\%} in the best case, without inflating agent output ($+$1.7\% mean, $-$0.3\% median).
\textsc{SkillReducer} requires ${\sim}$20--40 LLM calls per skill; at the average rate of commonly-used low-cost open-source models (${\sim}$\$0.3/M input tokens), compressing 600 skills costs ${\sim}$\$14--18, a one-time expense amortized within a few hundred invocations per skill.

\finding{RQ1: \textsc{SkillReducer} achieves 48\% description compression and 39\% body compression, with end-to-end savings of 26.8\%. Compression scales to wild skills (77.5\% on 198 samples) and is economically viable (\$14--18 for 600 skills).}

%% ------------------------------------------------------------
\subsection{RQ2: Functional Quality}
\label{sec:rq_quality}

\begin{table}[t]
\centering
\caption{Gate 2 task evaluation (N=600).}
\label{tab:gate2}
\begin{tabular}{lcccc}
\toprule
& $D$ & $A$ & $C$ & pass rate \\
\midrule
All (600) & 0.684 & 0.722 & \textbf{0.742} & 86.0\% \\
Official (87) & 0.696 & 0.743 & \textbf{0.831} & --- \\
Community (464) & 0.691 & 0.713 & \textbf{0.724} & --- \\
Wild (49) & 0.619 & 0.758 & \textbf{0.773} & --- \\
Non-ceiling (481) & 0.606 & 0.653 & \textbf{0.678} & 82.5\% \\
\midrule
\multicolumn{5}{l}{C vs.\ A: $p{=}0.002$**, $d{=}0.107$, C$>$A 25.3\%, C$<$A 14.0\%} \\
\multicolumn{5}{l}{C vs.\ D: $p{<}0.001$***, $d{=}0.171$} \\
\bottomrule
\end{tabular}
\end{table}

\begin{table}[t]
\centering
\caption{Baseline comparison (N=50, same token budget). One-sided sign tests on non-tied pairs.}
\label{tab:baselines}
\begin{tabular}{lcccc}
\toprule
Method & Score & Retention & W/T/L & $p$ \\
\midrule
A (Original) & 0.921 & --- & & \\
\textbf{C (\textsc{SkillReducer})} & \textbf{0.909} & \textbf{0.949} & --- & --- \\
\midrule
P (LLMLingua) & 0.767 & 0.820 & 28/16/6 & $<$0.001 \\
L (LLM direct) & 0.866 & 0.918 & 19/26/5 & $0.003$ \\
T (Truncation) & 0.791 & 0.845 & 27/17/6 & $<$0.001 \\
R (Random sent.) & 0.694 & 0.750 & 38/10/2 & $<$0.001 \\
\bottomrule
\end{tabular}
\end{table}

\noindent\textbf{Routing Preservation (Stage 1).}
\label{sec:stage1_results}
To validate that compressed descriptions still trigger correctly, we deploy each skill with its compressed description and test whether the agent invokes it for the same queries (Section~\ref{sec:stage1_phase2}).
Table~\ref{tab:stage1_trigger} shows the results across all 600 skills: 252 (42.0\%) passed directly, 245 (40.8\%) required selective restore, and 39 (6.5\%) fell back to the original description, yielding \textbf{536/536 (100\%)} routing preservation among non-obsolete skills.
The high restore rate confirms the two-phase design: Phase~1's simulated oracle is fast but overly optimistic, and Phase~2 catches over a third of over-compressed descriptions.
The remaining 64 skills (10.7\%) did not trigger even with the original description, suggesting skill obsolescence (Section~\ref{sec:rq_crossmodel}).

\begin{table}[t]
\centering
\caption{Stage 1 real-trigger validation (N=600).}
\label{tab:stage1_trigger}
\begin{tabular}{lr}
\toprule
Metric & Value \\
\midrule
Direct pass (no recovery needed) & 252 (42.0\%) \\
Pass after selective restore & 245 (40.8\%) \\
Fallback to original description & 39 (6.5\%) \\
Obsolescence (original fails too) & 64 (10.7\%) \\
\midrule
\textbf{Total routing preserved} & \textbf{536/536 (100\%)} \\
Description compression (passed) & 65\% mean \\
\bottomrule
\end{tabular}
\end{table}

\noindent\textbf{Task Evaluation (Stage 2).}
\label{sec:gate2_results}
Table~\ref{tab:gate2} shows the Gate~2 results across all 600 skills.
The core result is that compression preserves functional quality: 86.0\% of skills (95\% CI: [83.1\%, 88.7\%]) have $\text{score}_C \geq \text{score}_A$, with 25.3\% improving and only 14.0\% regressing.
The advantage holds after excluding ceiling-effect skills (N=481, $p=0.010$).
Counter-intuitively, condition~C outperforms condition~A on average (0.742 vs.\ 0.722, $p=0.002$, $d=0.107$), suggesting a \emph{less-is-more} effect where non-essential content distracts the agent from core instructions~\cite{Liu2024LostMiddle,Shi2023Distracted}.
The effect scales with body length (+11.8pp for official skills vs.\ +1.1pp for community).
Among the 84 regression cases (14.0\%), half are attributable to skill obsolescence (the agent performs equally well without the skill), 17\% to evaluation noise, and only 33\% to true compression regression, accounting for only \textbf{4.7\%} of all skills.
The recurring failure mode is example-as-specification: examples that implicitly define expected behavior are moved to reference modules.

\noindent\textbf{Baseline Comparison.}
\label{sec:baselines}
We compare against four baselines at the same token budget (Table~\ref{tab:baselines}):
LLMLingua~\cite{Jiang2023LLMLingua} (perplexity-based pruning),
LLM direct compression,
truncation, and
random sentence removal (50 skills, 3--5 tasks each).
\textsc{SkillReducer} significantly outperforms all four ($p \leq 0.003$).
LLMLingua's low retention (0.820) reflects a fundamental mismatch: perplexity-based pruning discards linguistically predictable but operationally critical tokens, while our taxonomy preserves all core rules.

\noindent\textbf{External Validation.}
\label{sec:skillsbench}
As an independent check free of LLM judges, we evaluate on SkillsBench~\cite{SkillsBench2026} (87 tasks, deterministic pytest verifiers): all 87 pass under both A and C with \textbf{zero regression}, though condition~D already passes 86/87, indicating a ceiling effect.
Our Gate~2 evaluation (where D scores only 0.684) better captures the value of skill content.

\finding{RQ2: Compression preserves quality (86.0\% pass rate, 87/87 on SkillsBench) and outperforms four baselines ($p \leq 0.003$). A less-is-more effect emerges where compression improves performance; only 4.7\% are true compression failures.}

%% ------------------------------------------------------------
\subsection{RQ3: Component Contribution}
\label{sec:rq_components}

\noindent\textbf{Ablation Study.}
\label{sec:ablation}
To quantify the contribution of each pipeline component, we evaluate six conditions on 50 skills (Table~\ref{tab:ablation}).
Taxonomy-driven classification is the key component: without it (C4), retention drops to 0.919, a 6.8pp gap versus the full pipeline.
Reference deduplication is the safest component: C3 achieves retention 1.000 and even outperforms the original ($0.944 > 0.939$), since it keeps the full body untouched and only removes redundant reference content, a pure less-is-more gain.
Description compression (C1) has negligible effect on task performance ($\Delta_A = -0.014$), as it operates on the routing layer.

\begin{table}[t]
\centering
\caption{Ablation study (50 skills $\times$ 5 tasks).}
\label{tab:ablation}
\begin{tabular}{llcccc}
\toprule
Cond. & Components & Score & Ret. & $\Delta_A$ \\
\midrule
D & (no skill) & 0.559 & 0.596 & $-$0.380 \\
A & Original & 0.939 & 1.000 & --- \\
\midrule
C & Full pipeline & \textbf{0.926} & \textbf{0.987} & $-$0.012 \\
C1 & Desc.\ only & 0.925 & 0.985 & $-$0.014 \\
C2 & Body compress (no classify) & 0.882 & 0.940 & $-$0.056 \\
C3 & Ref deduplication only & \textbf{0.944} & \textbf{1.000} & +0.005 \\
C4 & Compress all (no classify) & 0.863 & 0.919 & $-$0.076 \\
\bottomrule
\end{tabular}
\end{table}

The Gate~2 feedback loop also contributes to quality: 38 of 600 skills (6.3\%) triggered at least one feedback iteration; 22 recovered in the first round, and 9 of the remaining 16 recovered in the second (31/38 = 81.6\% total recovery), suggesting diminishing returns beyond the first iteration.
These feedback-triggering skills are systematically harder: longer bodies (2,990 vs.\ 1,673 tokens) and lower core-rule proportion (43.2\% vs.\ 51.1\%).

\noindent\textbf{Sensitivity to Compression Level.}
\label{sec:sensitivity}
Table~\ref{tab:sensitivity} breaks down retention by compression ratio.
Retention stays in the range 0.910--0.956 from light ($<$20\%) to aggressive ($>$80\%) compression, with pass rates consistently above 80\%, confirming that aggressive compression does not systematically increase risk.

\begin{table}[t]
\centering
\caption{Sensitivity analysis: retention by compression ratio.}
\label{tab:sensitivity}
\begin{tabular}{lccccc}
\toprule
Compression & N & Retention & Pass & C$>$A & C$<$A \\
\midrule
$<$20\% & 248 & 0.956 & 87.5\% & 34 & 31 \\
20--40\% & 58 & 0.939 & 84.5\% & 15 & 9 \\
40--60\% & 86 & 0.910 & 81.4\% & 25 & 16 \\
60--80\% & 126 & 0.931 & 85.7\% & 40 & 18 \\
$>$80\% & 82 & 0.942 & 87.8\% & 30 & 10 \\
\bottomrule
\end{tabular}
\end{table}

\finding{RQ3: Taxonomy classification is the key component (6.8pp gap without it). The feedback loop recovers 82\% of failing skills. Retention is stable across all compression levels.}

%% ------------------------------------------------------------

\subsection{RQ4: Cross-Model and Cross-Framework Generalization}
\label{sec:rq_crossmodel}

To test whether compression benefits transfer across models, we evaluate 30 skills on five models spanning four model families: Qwen3-max~\cite{Qwen2025Qwen3} (strong, proprietary), DeepSeek-V3~\cite{DeepSeek2024V3} (the compression model), Qwen2.5-7B~\cite{Qwen2024Qwen25} (weak, 7B parameters), GLM-5~\cite{GLM2024} (Zhipu family), and GPT-OSS-120B (open-source, 120B parameters).
All skills were compressed using DeepSeek-V3; the evaluation uses the full condition~C setup with \texttt{read\_file} tools for reference loading (flat injection for GPT-OSS-120B, which lacks tool-calling support), and 5 tasks per skill.

Table~\ref{tab:crossmodel} shows retention ranging from 0.939 (Qwen2.5-7B) to 0.986 (GLM-5), with a mean of \textbf{0.965} across four model families.
The weak model (Qwen2.5-7B) achieves retention 0.939, confirming that compressions produced by a strong model transfer to substantially smaller models.
GPT-OSS-120B achieves retention 0.982 despite using flat injection (no tool-calling support), indicating compression is robust even without progressive disclosure.
DeepSeek-V3 (the compression model) achieves 0.978.

\noindent\textbf{Cross-Compressor Validation.}
Table~\ref{tab:crossmodel} fixes the DeepSeek-V3 \emph{compressor} and varies the \emph{evaluator}; here we fix the evaluator and vary the compressor.
We re-compress the same 30 skills using Qwen3-max and Qwen2.5-7B instead of DeepSeek-V3.
Retention remains stable: 0.897 (Qwen3-max) and 0.874 (Qwen2.5-7B) vs.\ 0.896 (DeepSeek-V3), indicating that the pipeline's structure-aware design, not the compression model's capability, drives the quality preservation.

\noindent\textbf{Cross-Framework Validation.}
\label{sec:crossframework}
To verify that compressed skills transfer across agent frame\-works, not just models, we deploy 30 skills on OpenCode~(v1.2.27), an open-source coding agent with 120K+ GitHub stars, using its native skill mechanism.
OpenCode differs from our pipeline: it has its own system prompt, tool registry, and skill routing via a dedicated \texttt{skill} tool.
We use DeepSeek-V3 as the backend and evaluate 3 tasks per skill under conditions~A and~C. Results on 30 skills show mean retention of \textbf{0.944}, consistent with the 0.949 observed in our main baseline comparison (Table~\ref{tab:baselines}).
Compressed skills score 0.764 vs.\ originals at 0.751, again exhibiting the less-is-more effect.

\begin{table}[t]
\centering
\caption{Cross-model evaluation (N=30 skills, 5 tasks each).}
\label{tab:crossmodel}
\begin{tabular}{lccc}
\toprule
Model & score$_A$ & score$_C$ & retention \\
\midrule
GLM-5 & 0.926 & 0.962 & 0.986 \\
DeepSeek-V3 & 0.936 & 0.963 & 0.978 \\
Qwen3-max & 0.931 & 0.930 & 0.955 \\
GPT-OSS-120B & 0.951 & 0.958 & 0.982 \\
Qwen2.5-7B & 0.920 & 0.893 & 0.939 \\
\midrule
\textit{Mean} & \textit{0.933} & \textit{0.941} & \textit{0.965} \\
\bottomrule
\end{tabular}
\end{table}

\noindent\textbf{Skill Obsolescence.}
As LLMs grow more capable, skills may become redundant: 10.7\% of skills in Stage~1 did not trigger even with original descriptions, and condition~D achieves 98.9\% on SkillsBench.
Yet skills remain essential for weaker models: Qwen2.5-7B shows the largest gap between conditions A and D, suggesting that smaller models benefit most from skill content.
This tension motivates adaptive skill lifecycle management: Gate~2 can naturally identify candidates for retirement.

\finding{RQ4: Compression transfers across four model families (mean retention 0.965) and to an independent agent framework (OpenCode, retention 0.944). Model obsolescence affects 10.7\% of skills, motivating adaptive skill lifecycle management.}

%% ------------------------------------------------------------
%% ============================================================
\section{Related Work}
\label{sec:related}

\noindent\textbf{Prompt Compression and Context Management.}
LLMLingua~\cite{Jiang2023LLMLingua,Jiang2024LongLLMLingua} and Selective Context~\cite{Li2023SelectiveContext} prune low-perplexity tokens, Gisting~\cite{Mu2024Gisting} learns soft token embeddings, and RECOMP~\cite{Xu2023RECOMP} trains a compressor for retrieval-augmented generation; see~\cite{Li2025PromptSurvey} for a comprehensive survey.
Shi et al.~\cite{Shi2024ConcisePrecise} compress tool documentation (16$\times$) but require task-specific training data.
These methods treat the input as a flat token sequence, which is effective for homogeneous text but cannot distinguish actionable rules from supplementary examples within a skill.
On the context management side, RAG~\cite{Lewis2020RAG} and planning frameworks~\cite{Huang2024PlanSurvey} manage what enters the context window, while AgentCompress~\cite{Huang2025AgentCompress} compresses agent interaction histories.
Levy et al.~\cite{Levy2025ContextLength} show that context inflation degrades LLM performance by 13--85\%, providing independent support for our less-is-more finding.
\textsc{SkillReducer} is training-free and adapts classical SE techniques, including software debloating~\cite{Quach2018Debloating}, delta debugging~\cite{Zeller2002,Wang2025DDOR}, and program slicing~\cite{Weiser1984}, to exploit the internal structure of skills, applying type-specific compression while reducing the initial context footprint and preserving access to full content on demand.

\noindent\textbf{Tool and Skill Ecosystems.}
ToolBench~\cite{Qin2024ToolBench}, API-Bank~\cite{Li2023APIBank}, and AgentBench~\cite{Liu2024AgentBench} benchmark tool selection and agent capabilities, while Yang et al.~\cite{Yang2026SkillsWild} analyze 31K skills from a security perspective.
APE~\cite{Zhou2023APE} and CodeAct~\cite{Wang2024CodeAct} optimize prompts at authoring time.
In contrast, \textsc{SkillReducer} operates post-authoring as a build-time pass, complementing authoring-time and runtime approaches.

%% ============================================================
\section{Threats to Validity}
\label{sec:threats}

\noindent\textbf{Internal Validity.}
LLM-based classification introduces randomness; we mitigate this with temperature 0, conservative fallbacks, and the feedback loop.
Gate~2 serves dual roles (feedback signal and evaluation criterion), so the 86.0\% pass rate reflects optimization-to-criterion; the independent SkillsBench (87/87) partially mitigates this.
We use separate models for compression (DeepSeek-V3) and evaluation (Qwen3.5~\cite{Qwen2025Qwen3}), and validate the LLM judge against deterministic code-execution verifiers ($\kappa{=}0.939$ on 1,646 tasks).

\noindent\textbf{External Validity.}
All 600 skills use the skill protocol proposed by Anthropic~\cite{ClaudeCode2025}; generalization to other platforms (Cursor, Windsurf) remains untested.
The cross-model evaluation covers five models from four families on 30 skills; additional families (e.g., Llama) remain untested.
A stratified study on 198 wild skills confirms compression scalability but lacks Gate~2 evaluation due to cost.

%% ============================================================
\section{Conclusion}
\label{sec:conclusion}

We studied 55,315 LLM agent skills and found systemic token inefficiencies across descriptions, bodies, and references.
\textsc{SkillReducer} addresses this through delta-debugging-based routing compression (48\%) and taxonomy-driven progressive disclosure (39\% body reduction), achieving an 86\% pass rate across 600 skills with a less-is-more effect where compression improves performance.
Results transfer across five models from four families (retention 0.965) and to an independent agent framework (retention 0.944).
Only 4.7\% of skills are true compression failures; the remainder stem from skill obsolescence, a broader ecosystem issue motivating model-adaptive skill lifecycle management.
We will release \textsc{SkillReducer} as an open-source build-time tool for skill authors and the broader agent community.

\bibliographystyle{IEEEtran}
\bibliography{references}

%% ============================================================
%% APPENDIX (Supplementary Material)
%% ============================================================
\appendices

\section{Running Example: \texttt{marketing-strategy-pmm}}
\label{app:example}

We illustrate the full pipeline on the \texttt{marketing-strategy-pmm} skill (product marketing, positioning, GTM strategy).

\noindent\textbf{Stage~1: Description.} (87 $\to$ 32 tokens, 63\% reduction).
The original description enumerates every feature and trigger phrase: \emph{``Product marketing, positioning, GTM strategy, competitive intelligence{\ldots}Use when developing positioning, planning product launches, creating messaging, analyzing competitors, entering new markets, enabling sales, or when user mentions product marketing, positioning, GTM, go-to-market{\ldots}''}.
\textsc{ddmin} identifies that the trigger-phrase list is redundant because the router can infer triggers from the feature keywords alone.
The 1-minimal result: \emph{``Product marketing, positioning, GTM strategy, competitive intelligence. Tools: ICP definition, April Dunford methodology, launch playbooks, battlecards, market entry guides''}.

\noindent\textbf{Stage~2: Body.} (2,543 $\to$ 540 tokens, 79\% reduction).
The classifier identifies table-of-contents entries and workflow explanations as background, persona-specific messaging as examples, and HubSpot configuration snippets as templates.
The remaining core (KPIs, methodology steps, decision criteria) is compressed into concise bullet points.
Three body-derived reference modules are created with routing metadata:
\texttt{templates.md} (327 tok, when: ``you need to WRITE HubSpot configs''),
\texttt{examples.md} (684 tok, when: ``you need to SEE persona messaging''),
\texttt{background.md} (602 tok, when: ``you need to understand positioning methodologies'').

\noindent\textbf{Stage~2: Original References.} (9,476 $\to$ 6,691 tokens, 29\% reduction).
Four external reference files are deduplicated against the core and annotated with routing metadata.

\noindent\textbf{End-to-End.} Invocation cost drops from 12,019 to 540 tokens for core-only tasks (96\%), or to 7,231 tokens if all references are loaded (40\%).
Gate~2 confirms $\text{score}_C{=}1.0$ vs.\ $\text{score}_A{=}0.93$.

Table~\ref{tab:example} summarizes the end-to-end compression.

\begin{table}[H]
\centering
\caption{Before/after compression of the \texttt{marketing-strategy-pmm} skill.}
\label{tab:example}
\begin{tabular}{lrrr}
\toprule
Component & Original & Compressed & Reduction \\
\midrule
Description & 87 tok & 32 tok & 63\% \\
Body (always loaded) & 2,543 tok & 540 tok & 79\% \\
References (on demand) & 9,476 tok & 6,691 tok & 29\% \\
\midrule
Total if all loaded & 12,019 tok & 7,231 tok & 40\% \\
Core-only invocation & 2,543 tok & 540 tok & 79\% \\
\bottomrule
\end{tabular}
\end{table}

\section{GMM Model Selection}
\label{app:gmm}

Figure~\ref{fig:b2_elbow} shows the silhouette score analysis used to select $k=5$ for the GMM clustering. The silhouette score reaches its first local peak at $k=5$ (0.393), matching our five-category taxonomy; higher values of $k$ do not consistently improve separation.

\begin{figure}[H]
  \centering
  \includegraphics[width=\linewidth]{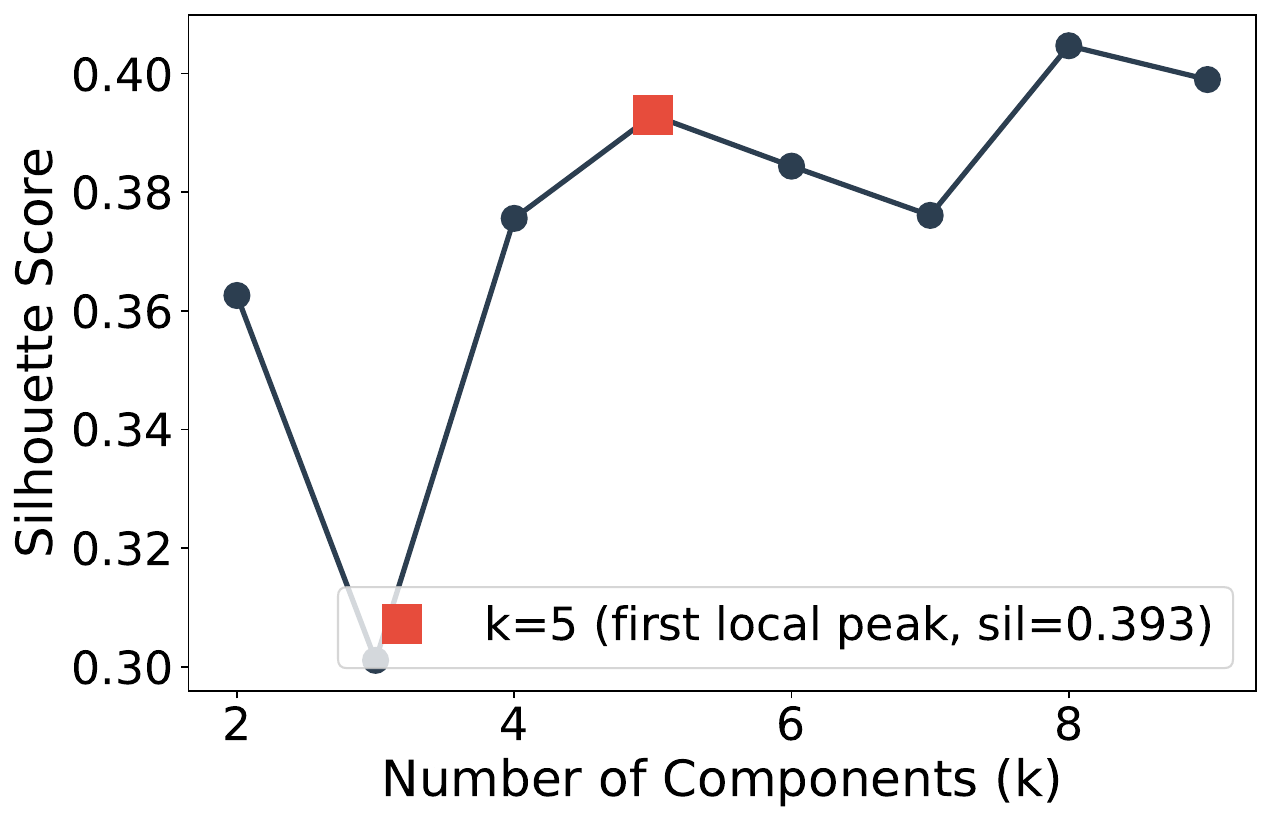}
  \caption{GMM model selection. $k=5$ is the first local peak (silhouette = 0.393), matching the five-category taxonomy.}
  \label{fig:b2_elbow}
\end{figure}

\section{Dataset Statistics}
\label{app:datasets}

Table~\ref{tab:datasets} provides detailed statistics of all datasets used in this work.

\begin{table}[H]
\centering
\caption{Dataset statistics.}
\label{tab:datasets}
\begin{tabular}{lrrr}
\toprule
Dataset & $N$ & Body (tok) & Desc (tok) \\
\midrule
Wild & 55,315 & 2,796 & 45.9 \\
SkillHub & 100 & 5,922 & 47.78 \\
Community & 620 & 1,400 & 27.21 \\
\midrule
Evaluation (Official) & 87 & --- & --- \\
Evaluation (Community) & 464 & --- & --- \\
SkillsBench & 87 tasks & --- & --- \\
\bottomrule
\end{tabular}
\end{table}

\section{Quality Safeguards}
\label{app:guarantee}

\textsc{SkillReducer}'s pipeline is designed to maximize the likelihood that $\text{score}_C \geq \text{score}_A$ through three mechanisms acting in sequence:

\begin{enumerate}
  \item \textbf{Gate 1 (Faithfulness).} After compression, an LLM verifies that all key operational concepts from the original body are preserved in the compressed core ($\mathcal{C}(s.b) \subseteq \mathcal{C}(b^*)$). If any concept is missing, the compression is rolled back to the original for that content type.

  \item \textbf{Gate 2 + Feedback Loop.} If $\text{score}_C < \text{score}_A$ on any task, the system identifies which non-core items caused the failure and promotes them to core. The body is recompressed and retested (up to two iterations). This monotonically increases $|b^*|$ until task performance recovers.

  \item \textbf{Fallback.} If the feedback loop fails to restore $\text{score}_C \geq \text{score}_A$ after two iterations, the current best compression is kept (i.e., the version with the most promoted items), which represents the closest approximation to the original.
\end{enumerate}

This means the only case where $\text{score}_C < \text{score}_A$ can occur in practice is when the LLM-generated evaluation tasks in Gate~2 do not cover the failure mode that manifests at deployment time, a limitation of any finite test suite, analogous to the incompleteness of software testing.

In our experiments, across 600 skills, 86.0\% achieve $\text{score}_C \geq \text{score}_A$.
The 14.0\% with regression are typically skills with deeply interleaved content where examples serve as implicit rules; the classifier moves them to reference modules, but they are needed in the always-loaded core.
This represents a fundamental limitation of taxonomy-based classification: when authors intentionally embed rules within examples, automated separation can lose implicit dependencies.

\section{Theoretical Analysis}
\label{app:theory}

We formalize two properties of \textsc{SkillReducer}: the convergence guarantee of the feedback loop, and the expected token cost under progressive disclosure.

\subsection{Feedback Loop Convergence}

\begin{proposition}[Monotone Convergence of Gate~2 Feedback]
\label{thm:convergence}
Let $\mathcal{I} = \{x_1, \ldots, x_n\}$ be the set of all items in a skill body, and let $I_{\textup{core}}^{(0)} \subseteq \mathcal{I}$ be the initial core set produced by the classifier.
At each feedback iteration $i$, the promotion operator $\mathcal{P}$ adds items from $\mathcal{I} \setminus I_{\textup{core}}^{(i)}$ to form $I_{\textup{core}}^{(i+1)} = I_{\textup{core}}^{(i)} \cup \mathcal{P}^{(i)}$, where $\mathcal{P}^{(i)} \neq \emptyset$ whenever Gate~2 fails.
Then:
\begin{enumerate}
  \item[(i)] The core set is monotonically non-decreasing: $I_{\textup{core}}^{(0)} \subseteq I_{\textup{core}}^{(1)} \subseteq \cdots$
  \item[(ii)] The loop terminates in at most $|\mathcal{I}| - |I_{\textup{core}}^{(0)}|$ iterations.
  \item[(iii)] If Gate~2 is monotone in coverage (adding items never decreases $\textup{score}_C$), then $\textup{score}_C^{(i+1)} \geq \textup{score}_C^{(i)}$.
\end{enumerate}
\end{proposition}

\begin{proof}
(i)~By construction, $\mathcal{P}^{(i)} \subseteq \mathcal{I} \setminus I_{\textup{core}}^{(i)}$, so $I_{\textup{core}}^{(i+1)} \supseteq I_{\textup{core}}^{(i)}$.

(ii)~Since $|I_{\textup{core}}^{(i+1)}| > |I_{\textup{core}}^{(i)}|$ when $\mathcal{P}^{(i)} \neq \emptyset$, and $|I_{\textup{core}}^{(i)}| \leq |\mathcal{I}|$, the loop terminates after at most $|\mathcal{I}| - |I_{\textup{core}}^{(0)}|$ iterations.
In practice, we cap at $K=2$ iterations.

(iii)~The monotonicity assumption states that providing more skill content to an agent cannot decrease task performance, i.e., $I \subseteq I' \implies \textup{score}_C(I') \geq \textup{score}_C(I)$.
Under this assumption, since $I_{\textup{core}}^{(i)} \subseteq I_{\textup{core}}^{(i+1)}$, we have $\textup{score}_C^{(i+1)} \geq \textup{score}_C^{(i)}$.

Note that the monotonicity assumption in~(iii) does not always hold in practice: adding content can distract the agent (our less-is-more finding).
However, the promoted items $\mathcal{P}^{(i)}$ are specifically selected to address failed rubric criteria, making distraction unlikely for these targeted additions.
Empirically, the loop improves retention for the majority of triggered skills: 71.1\% (27/38) achieve $\textup{score}_C \geq \textup{score}_A$ after at most two iterations.
The remaining 28.9\% are cases where (iii)'s monotonicity assumption is violated; the promoted items introduce distraction that offsets their informational benefit.
\end{proof}

\subsection{Expected Cost Under Progressive Disclosure}

\begin{proposition}[Expected Invocation Cost]
\label{prop:cost}
Let a skill $s$ have body $s.b$ with $|\mathcal{I}|$ items, of which a fraction $\rho$ are classified as core.
After compression, the core module has $\rho \cdot |s.b|$ tokens (before core-internal compression) and $k$ reference modules $R^* = \{r_1^*, \ldots, r_k^*\}$.
If each reference is loaded independently with probability $p_j$, the expected invocation cost is:
\begin{equation}
  \label{eq:expected_cost}
  \mathbb{E}[\textup{Cost}'(s)] = |s'.d| + \alpha \cdot \rho \cdot |s.b| + \sum_{j=1}^{k} p_j \cdot |r_j^*|
\end{equation}
where $\alpha \in (0, 1]$ is the core-internal compression factor (merging and tightening of core rules).
\end{proposition}

\begin{proof}
The description $s'.d$ and compressed core $\alpha \cdot \rho \cdot |s.b|$ are always loaded.
Each reference $r_j^*$ is loaded only when the agent issues a \texttt{read\_file} call, which occurs with probability $p_j$.
By linearity of expectation, the total expected cost follows.
\end{proof}

\noindent\textbf{Instantiation with Empirical Values.}
From our evaluation:
$\rho = 0.383$,
$\alpha \approx 0.63$ (core-internal compression from Gate~1 data),
$k \approx 2.1$ reference modules per skill (mean),
and $p \approx 0.30$ (mean reference loading rate from Gate~2 tool-call logs).
Assuming uniform $p_j = p$ and mean reference size $|r_j^*| \approx 0.617 \cdot |s.b| / k$ (the non-core content, evenly split):

\begin{align}
  \mathbb{E}[\text{Cost}'(s)] &\approx |s'.d| + (0.63 \times 0.383 + 0.30 \times 0.617) \cdot |s.b| \notag \\
  &= |s'.d| + (0.241 + 0.185) \cdot |s.b| \notag \\
  &= |s'.d| + 0.426 \cdot |s.b|
\end{align}

This yields a \textbf{57.4\%} expected body cost reduction in the best case (description cost is negligible: mean 45.9 tokens vs.\ body mean 2,796 tokens).
In a conservative scenario ($\alpha = 1.0$, no core-internal compression; $p = 0.5$, high reference usage), the reduction drops to $1 - (0.383 + 0.5 \times 0.617) = 30.9\%$.
Our empirical measurement of 26.8\%--43.2\% falls within this range, with the variance driven by heterogeneous skill lengths, compression ratios, and reference loading patterns across the 490 analyzed skills.

\end{document}